\documentclass[onecolumn,10pt]{IEEEtran}

\usepackage{amsfonts,amsmath,amssymb,amsthm,caption}
\usepackage{amsbsy}
\usepackage{graphicx,multirow,bm}
\usepackage{color}
\usepackage{algorithm}
\usepackage{algpseudocode}
\makeatletter
\newtheoremstyle{mythm}{3pt}{3pt}{}{16pt}{\bfseries}{:}{.5em}{}
\theoremstyle{mythm}
\newtheorem{theorem}{Theorem}
\newtheorem{example}{Example}
\newtheorem{definition}{Definition}
\newtheorem{remark}{Remark}

\newtheorem{corollary}{Corollary}
\newtheorem{lemma}{Lemma}
\newtheorem{construction}{Construction}

\newcommand{\tabcaption}{\def\@captype{table}\caption}

\begin{document}

\title{Improved Constructions of Coded Caching Schemes for Combination Networks
\author{Minquan Cheng, Yiqun Li, Xi Zhong, Ruizhong Wei
}
\thanks{M. Cheng, Y. Li and X. Zhong are with Guangxi Key Lab of Multi-source Information Mining $\&$ Security, Guangxi Normal University,
Guilin 541004, China, (e-mail: chengqinshi@hotmail.com).} 
\thanks{R. Wei is with Department of Computer Science, Lakehead University, Thunder Bay, ON, Canada, P7B 5E1,(e-mail: rwei@lakeheadu.ca).}
}
\date{}
\maketitle

\begin{abstract}
In an $(H,r)$ combination network, a single content library is delivered to  ${H\choose r}$ users through deployed $H$ relays without cache memories, such that each user with local cache memories is simultaneously served by a different subset of $r$ relays on orthogonal non-interfering and error-free channels. The combinatorial placement delivery array (CPDA in short) can be used to realize a coded caching scheme for combination networks. In this paper, a new algorithm realizing a coded caching scheme for combination network based on a CPDA is proposed such that the schemes obtained have smaller subpacketization levels or are implemented more flexible than the previously known schemes.  Then we focus on directly constructing CPDAs for any positive integers $H$ and $r$ with $r<H$. This is different from the grouping method in reference (IEEE ISIT, 17-22, 2018) under the constraint that $r$ divides $H$. Consequently two classes of CPDAs are obtained. Finally comparing to the schemes and the method proposed by Yan et al., (IEEE ISIT, 17-22, 2018) the schemes realized by our CPDAs have significantly advantages on the subpacketization levels and the transmission rates.

\end{abstract}

\section{Introduction}
\label{Introduction}
It is well known that coded caching which proposed by Maddah-Ali and Niesen in \cite{MN} is an effective strategy to reduce congestion by placing some content in user local memories during off-peak hours with the hope that the pre-fetched content will generate more coded gain to reduce the number of broadcast transmissions from the server to the users during peak hours. Now coded caching has been widely used to various settings, Gaussian broadcast channels \cite{BWY}, multi-antenna fading channels \cite{NYK,SCK,ZE}, insecure channel \cite{STCD}, D2D networks \cite{JCM},  hierarchical networks \cite{KNAD}  ect.

An idealized version of a hierarchical network, i.e., combination network, has been  widely studied recently such as \cite{JWTLCEL,JTLC,WJPT,WJPT1,Yan,ZY,ZY1,SRDS,WTPJ,SRDSC,WTJP} and so on. In a $(H,r)$ combination network shown in Fig. \ref{fig.1}, a server having $N$ equal size files connects to $H$ relay nodes through independent orthogonal and noiseless links. The relay nodes are connected to $K={H\choose r}$ users in the following way: each user is connected to a distinct subset of $r$ relay nodes. This implies that a content from the server is delivered to the users through the relays, such that each user is simultaneously served by $r$ relays via orthogonal non-interfering channels. Here each relay $h$ just broadcasts the intermediate signals from the server to the users, which are associated with it.
Assume that each user has a storage capacity of $M$ files. A $(H,r,M,N)$ coded caching scheme consists of two independent phases: placement phase, which determines the contents to be cached by each users, and delivery phase, which determines the coded signal to be transmitted over each link for every possible requirement. The maximum amount of the transmission from the server to each relay $h$ in the delivery phase is called the rate $R_h$ which should be as small as possible.
\begin{figure}[h]
  \centering
  \includegraphics[scale=0.4]{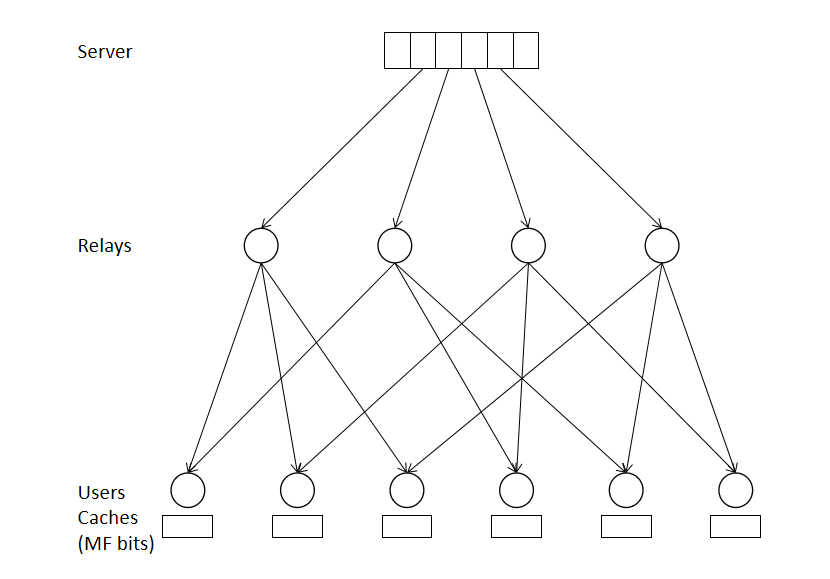}
  \caption{A combination network with $H=4$, $r=2$}\label{fig.1}
\end{figure}

The first study was proposed by Ji et al. in \cite{JWTLCEL,JTLC} for the case when $r$ divides $H$ (denote by $r|H$). \cite{WJPT} studied the case derived the lower bound on the rate by means of cyclic index coding outer bound. \cite{WJPT1} proposed a scheme such that the maximum coded gains from both the users and the relays are obtained. Up to now, the rate of this scheme is minimum among the previously known schemes for small $M$. However the above discussions are proposed under the assumption that each file in the library can be divided into any large equal packets, i.e., the packet number $F$ of each file divided grows exponentially with user number $K$. Clearly this limitation is critical in practice. Furthermore, the packet number $F$ represents the complexity of the scheme, i.e., the complexity of a scheme increases when $F$ increases. According to the structure of the combination network, the authors in \cite{Yan} proposed a concept of combinatorial placement delivery array (CPDA in short) which is a expanding the concept of PDA and can be used to realize a coded caching scheme for combination network. When $r|H$, by grouping method, the CPDAs can be obtained based on PDAs. Consequently the related coded caching schemes for the $(H,r)$ combination network are obtained. The main result in \cite{Yan} is that for any positive integers $H$, $r$, $M$ and $N$ satisfying $r|H$ and $\frac{M}{N}\in\{\frac{1}{K_1}$, $\frac{2}{K_1}$, $\ldots$, $\frac{K_1-1}{K_1}\}$ where $K_1={H-r\choose r-1}$, there exists a scheme with the rate $\frac{{H\choose r}}{1+K_1\frac{M}{N}}\left(1-\frac{M}{N}\right)$ and the packet number $r{K_1\choose K_1\frac{M}{N}}$ for a $(H,r)$ combination network. Yan et al. in \cite{Yan} pointed out that the packet number of this scheme is much smaller than the packet number of the scheme in \cite{WJPT1}. It is  not difficult to check the packet number in this scheme still grows exponentially with user number $K={H\choose r}$. Furthermore, the schemes in \cite{Yan} abandoned many coded gains, which could be generated among the users, due to the grouping method. It is well known that
the value of coded gains are very important for reducing the rate.

There are some other studies with additional conditions for combination networks. For example, references \cite{ZY1,WTPJ,ZY,SRDSC} studied caching schemes for the combination network equipped with caches at both the relay nodes and the end users by means of Maximum Distance Separable (MDS) codes and interference alignment.

In this paper we consider the schemes for a $(H,r)$ combination network when $N\geq K$ by means of constructing CPDAs. Due to the deficiencies of the grouping method, we first propose a new algorithm which can be used to realize a coded caching scheme for combination network based on a CPDA, and then directly construct some CPDAs such that the coded gain of the scheme realized by directly constructing CPDA is larger than that of the schemes realized by the CPDAs in \cite{Yan}. Different from the method of realizing a scheme based a CPDA in \cite{Yan}, given the same CPDA, the packet number of the scheme realized by our algorithm is smaller than that of the scheme in \cite{Yan}. Furthermore, each coded signal is completed the transmission by several relays simultaneously. This is also important in practice when the size of transmission data is very large at each time slot in our algorithm. As for the direct construction, we first show that the previously known PDAs in \cite{YTCC}, which is obtained by the results on strongly coloring in bipartite graph proposed by Jennifer \emph{et al.} in \cite{Jen1997subset}, are also CPDAs. As a generalization of above CPDA, a new class of CPDA is obtained. It is worth noting that our CPDAs hold for any positive integers $H$ and $r$ with $r<H$. By performance analysis, we show that the scheme realized by our CPDAs has much smaller packet number than that of the scheme in \cite{Yan} while the transmission rate increases just several times. Furthermore, the schemes realized by the above two classes of CPDAs have both  packet number and the transmission rate smaller than that of the schemes generated by grouping method proposed in \cite{Yan} based on the original PDAs in \cite{YTCC}. This implies that considering the direct construction of CPDAs is useful.

The rest of this paper is organized as follows. Section \ref{preliminaries} reviews the system model of combination network and the concept of a CPDA. Section \ref{sec-algorithm} introduce a new algorithm to realize a scheme based on a CPDA. In Section \ref{sec-new-scheme}, the fact that the PDA in \cite{YTCC} is a CPDA is shown and a new class of CPDAs is constructed. Our new scheme is then
proposed. In Section \ref{sec-peformance}, the comparisons between the schemes in \cite{Yan} and our new schemes are proposed. Finally conclusion is drawn in Section \ref{conclusion}.

\section{Preliminaries}
\label{preliminaries}
For any positive integer $H$ and $r$ with $r<H$, let $[H]=\{1,2,\ldots,H\}$ and let ${[H]\choose r}=\{A\ :\  A\subseteq [H], |A|=r\}$, i.e., ${[H]\choose r}$ is the collection of subsets of $[H]$ of size $r$.  We use $|\cdot|$ to represent the cardinality of a set; $A-B=\{x\ |\ x\in A, x\not\in B\}$ for any two sets $A$ and $B$.
\subsection{System Model}\label{sub-system}
In a $(H,r,M,N)$ combination network, a server has access to $N$ files denoted by $\mathcal{W}=\{W_1,W_2,\ldots,W_N\}$, each composed of $E$ independent and identically uniformly distributed (i.i.d.) random bits, and is connected to $H$ relays, say $\mathcal{H}=\{h_1,h_2,\ldots,h_H\}$, through $H$ error-free and interference-free links. All the relays have no memory. The relays are connected to $K={H\choose r}$ users, say $\mathcal{K}$, through $rK$ error-free and interference-free links.  The set of users connected to the relay $h\in\mathcal{H}$ is denoted by $\mathcal{U}_h$ and the set of relays connected to user $k$ is denoted by $\mathcal{H}_k$. So each user is always labeled by $\mathcal{H}_k$ in this paper. In combination network each subset $\mathcal{H}_k$ has exactly $r$ relays and is unique, i.e., $\mathcal{K}={\mathcal{H}\choose r}$. Assume that user $k\in \mathcal{K}$ stores information about the $N$ files in its cache of size $ME$ bits, where $M\in[N]$. We denote the content in the cache of user $k\in \mathcal{K}$ by $Z_k$ and the set of all the cached contents by $\mathcal{Z}\triangleq\{Z_k,k\in \mathcal{K}\}$. A $(H,r,M,N)$ coded caching scheme for combination network consists of two independent phases:
\begin{itemize}
\item Placement phase: During the off peak traffic times, each user $k$ directly accesses to the file library $\mathcal{W}$ and stores an arbitrary function
thereof in its cache memory, subject to the space limitation of $ME$ bits. That is, for each user $k\in \mathcal{K}$, there exists a function $\phi_k\ :\ \mathbb{F}^{NE}_2\longrightarrow \mathbb{F}^{ME}_2$ generates the cache contents $Z_k\triangleq \phi_k(W_n, n=1,2,\ldots,N)$.

\item Delivery phase: During the peak traffic times, assume that each user requests one file from the files library $\mathcal{W}$ randomly and independently. The request file number is denoted by $\mathbf{d}=(d_1,d_2,\cdots,d_{K})$, which indicates that user $k$ requests the $d_k$-th file $W_{d_k}$ for any $d_k\in [N]$ and $k\in \mathcal{K}$. The server sends a message $X_{S\rightarrow h}$ to relay $h\in\mathcal{H}$. Then relay $h$ transmits $X_{S\rightarrow h}$ to user $k\in \mathcal{U}_h$. User $k\in\mathcal{K}$ must recover its desired file $W_{d_k}$ from $Z_k$ and ($X_{S\rightarrow h}: h \in \mathcal{H}_k$). So this phase can be represented by a class of encoding functions and a class of decoding functions:
\begin{itemize}
\item The encoding function from the server to relay $h\in\mathcal{H}$
$$\psi_{S\rightarrow h}\ :\  \mathbb{F}^{NE}_2\times\mathbb{F}^{KME}_2\times[N]^K\rightarrow \mathbb{F}^{L_{S\rightarrow h}}_2$$
generates the transmitted message $X_{S\rightarrow h}\triangleq \psi_{S\rightarrow h}(\mathcal{W},\mathcal{Z},{\bf d})$ as a function of the library $\mathcal{W}$, the cached information of all users $\mathcal{Z}$ and the demand vector ${\bf d}$, where $L_{S\rightarrow h}$ is the size of $X_{S\rightarrow h}$.

\item The decoding function $$\mu_k\ :\ \mathbb{F}^{(L_{S\rightarrow h}\ :\ h\in \mathcal{H}_k)}_2\times \mathbb{F}^{ME}_2\times [N]^K\rightarrow \mathbb{F}^{E}_2$$ decodes the request of user $k$ from all messages received by $k$ and its own cache, i.e.,
    $$W_{d_k}=\mu_k(\{X_{S\rightarrow h}\ :\ h\in \mathcal{H}_k\}, Z_k,{\bf d}).$$
\end{itemize}
The maximum link load (equivalent to download
time) in terms of files for each relay $h$ is defined as
\begin{align*}
R_h=\max_{{\bf d}\in[N]^K}\left\{\  \frac{L_{S\rightarrow h}}{E}\ \right\}.
\end{align*}
\end{itemize}
$R_h$ is called the transmission rate for relay $h$. When a coded caching scheme is implemented, each file must be divided into certain packets. Denote the packet number of each file by $F$. Since the complexity of a coded caching scheme increases when the packet number $F$ increases, we prefer to design a scheme with the transmission rate for each relay and the packet number as small as possible while satisfying all user demands.
\subsection{Placement delivery array for Combination Network}
\begin{definition}(PDA, \cite{YCTC})\label{def1} For positive integers $K$, $F$, $Z$ and $S$, an $F\times K$ array $\mathbf{P}=(p_{j,k})$, $j\in [F]$, $k\in [K]$, composed of a specific symbol $"*"$ and $S$ ordinary symbols $1,\cdots,S$, is called a $(K,F,Z,S)$ placement delivery array (PDA), if it satisfies the following conditions:
\begin{itemize}
\item[C1.] The symbol $"*"$ appears $Z$ times in each column;
\item[C2.] For any two distinct entries $p_{j_1,k_1}$ and $p_{j_2,k_2}$, we have $p_{j_1,k_1}=p_{j_2,k_2}=s$, an ordinary symbol only if
\begin{itemize}
\item[a.] $j_1\neq j_2$, $k_1\neq k_2$, i.e., they lie in distinct rows and distinct columns; and
\item[b.] $p_{j_1,k_2}=p_{j_2,k_1}=*$, i.e., the corresponding $2\times 2$ sub-array formed by rows $j_1$, $j_2$ and columns $k_1$, $k_2$ must be of the following form
$$\left(\begin{array}{cc}
       s\ \ *\\
       *\ \ s\\
       \end{array}\right) or
       \left(\begin{array}{cc}
       *\ \ s\\
       s\ \ *\\
       \end{array}\right).$$
\end{itemize}
\end{itemize}
\end{definition}
The first PDA was proposed by Maddah-Ali and Niesen in \cite{MN}. Such a PDA is referred to MN PDA. Based on MN PDA, the following two variants of PDAs are obtained.
\begin{lemma}(\cite{MN,YCTC})\label{le-MN-C}
For any positive integers $k$ and $t$ with $0<t<k-1$, we have the following PDAs.
\begin{itemize}
\item MN PDA: $(t+1)$-$(k,{k\choose t},{k-1\choose t-1},{k\choose t+1})$ PDA \cite{MN,YCTC};
\item Variant of MN PDA: $({k\choose t},k,t,{k\choose t+1})$ PDA \cite{CYTJ};
\item  Variant of MN PDA: $({k\choose t+1},{k\choose t},{k\choose t}-(t+1),k)$ PDA \cite{CYTJ};
\end{itemize}
\end{lemma}
\begin{remark}
\label{remark-1}
The authors in \cite{CYTJ} showed that for the fixed values of $K$, $F$ and $Z$, all the PDAs in Lemma \ref{le-MN-C} have the minimum value of $S$. Furthermore for the fixed values of $K$, $\frac{Z}{F}$ and $\frac{S}{F}$, all the PDAs in Lemma \ref{le-MN-C} have the minimum value of $F$.
\end{remark}

Yan et al., in \cite{YTCC} showed that the results on the strongly coloring in bipartite graph proposed by Jennifer et al. in \cite{Jen1997subset} are equivalent to the following PDAs.
\begin{lemma}(\cite{YTCC})\label{le:subset} Let $H$, $r$, $b$, $\lambda$ be positive integers satisfying $0<r,b<H$, $\lambda\leq\min\{r,b\}$ and $r+b-2\lambda<H$. There exists a $(K,F,Z,S)$ PDA, where
\begin{align}
&K={H\choose r},~ F={H\choose b},~Z={H\choose b}-{r\choose \lambda}{H-r\choose b-\lambda},\notag\\
&S={H\choose r+b-2\lambda}\cdot\min\left\{{H-(r+b-2\lambda)\choose \lambda},{r+b-2\lambda\choose r-\lambda }\right\}.\notag
\end{align}
\end{lemma}
\begin{remark}
\label{remark-2}
The authors in \cite{Jen1997subset} showed that the number of strong coloring of the bipartite graph generated in \cite{Jen1997subset} is minimum. This implies that for the fixed values of $K$, $F$, $Z$ and the position of stars of the PDAs placed according to the bipartite graph in \cite{Jen1997subset} the PDA obtained by the bipartite graph in \cite{Jen1997subset} has the minimum value of $S$.
\end{remark}

From Subsection \ref{sub-system}, we know that in a $(H,r)$-combination networks each user can be represented by a unique $r$-subset of relays which could serve him/her. The authors in \cite{Yan} proposed a concept of combinatorial PDA (CPDA) which is a special PDA and can be used to realize a scheme for combination networks. Form \cite{YCTC} when we use a PDA to realize a coded caching scheme, each user is regarded as a column. In \cite{Yan} the authors also used this idea. From Subsection \ref{sub-system}, each user can be represented by a $r$-subset of $\mathcal{H}$. So each column of a CPDA is always represented by a $r$-subset of $\mathcal{H}$ in the following definition.
\begin{definition}(\cite{Yan})\label{def2}
 For positive integers $H$, $r$, $F$, $Z$ and $S$, let $K={H\choose r}$. An $F\times {H\choose r}$ array $\mathbf{P}=(p_{j,A})$, $j\in [F]$, $A\in {[H]\choose r}$, composed of a specific symbol $"*"$ and $S$ ordinary symbols $1,\cdots,S$, is called a $(K,F,Z,S)$ combinatorial placement delivery array (CPDA), if it satisfies conditions C1-2 in Definition \ref{def1} as well as the following condition:
\begin{itemize}
\item[C3.] for any ordinary symbol $s\in [S]$, the labels of all columns containing symbol $s$ have nonempty intersection, where the label of a column contains all the symbols
appeared in that column.
\end{itemize}
\end{definition}
Given a $({H\choose r},F,Z,S)$ CPDA, we can obtain a $(H,r,M,N)$ coded caching scheme by the following result.
\begin{lemma}(\cite{Yan})
\label{le-CPDA-CC}
Given a $({H\choose r},F,Z,S)$ CPDA, there exists a $(H,r,M,N)$ coded caching scheme such that $\frac{M}{N}=\frac{Z}{F}$, the packet number is $hF$ and the transmission rate for each relay is $R_h=\frac{S}{HF}$.
\end{lemma}

In Theorem 1 in \cite{Yan} (i.e., the above Lemma \ref{le-CPDA-CC}) the authors said that the packet number is not larger than $HF$. However from their proof, we can check that the packet number in their Theorem 1 is $HF$. When $r|H$, the authors in \cite{Yan} proposed some CPDAs by grouping method based on PDAs.
\begin{lemma}(\cite{Yan})
\label{le-yan-recursive}
Let $H$, $r$ and $t$ be positive integers with $r|h$. If there exists a $(K_1,F,Z,S)$ PDA with $K_1={H-1\choose r-1}$, a $(K=\frac{H}{r}K_1, rF,rZ,HS)$ CPDA can be obtained by {\bf Transformation 1} in \cite{Yan}.
\end{lemma}
The authors showed that the third PDA in Lemma \ref{le-MN-C} is also a CPDA. That is the following result.
\begin{lemma}(\cite{Yan})
\label{le-yan-2}
There exists a $({H\choose r}, {H\choose r-1}, {H\choose r-1}-r,H)$ -CPDA which gives a $(K,M,N)$ coded caching scheme for $(H,r)$ combination networks with $\frac{M}{N}=1-\frac{r}{{H\choose r-1}}$ and $R_h=\frac{1}{{H\choose r-1}}$ with packet number $F={H\choose r-1}$.
\end{lemma}
Based on MN PDA, the main result in \cite{Yan} can be obtained as follows using Lemma \ref{le-yan-recursive}.
\begin{corollary}(\cite{Yan})
\label{co-yan-1}
Let $H$, $r$ and $t$ be positive integers with $r|H$ and $t<{H-1\choose r-1}$. There exists a $\left({H\choose r}\right.$, $r{{H-1\choose r-1}\choose t}$, $r{{H-1\choose r-1}-1\choose t-1}$,$\left. H{{H-1\choose r-1}\choose t+1}\right)$ PDA which gives a $(H,r,M,N)$ coded caching scheme with memory ratio $\frac{M}{N}=\frac{t}{{H\choose r-1}}$, transmission rate for each relay $R_{h}=\frac{K(1-M/N)}{H(1+K_1M/N)}$ and packet number $F=r{{H-1\choose r-1}\choose t}$, $h\in\mathcal{H}$.
\end{corollary}

\section{A new algorithm realizing schemes based on CPDAs}
\label{sec-algorithm}
In fact, the CPDAs in \cite{Yan} have a stronger structure since they are obtained by grouping method based on some PDA, i.e. the authors obtained CPDAs by means of {\bf Transformation 1} in \cite{Yan}. Consequently there are some limitations in their schemes. For examples, $r$ must divide $H$ and each packet has to be further divided into a $H$ subpackets with equal size. In fact, given a CPDA, sometimes we can deal with each packet more flexibly. So in this paper we will use the following algorithm, i.e., Algorithm \ref{alg:CPDA} to realize a $(H,r,M,N)$ coded caching scheme based on a $(K,F,Z,S)$ CPDA.
\begin{algorithm}[http!]
\caption{Delivery phase of a caching scheme based on CPDA}\label{alg:CPDA}
\begin{algorithmic}[1]
\Procedure {Placement}{$\mathbf{P}$, $\mathcal{W}$}
\State Split each file $W_i\in\mathcal{W}$ into $F$ packets, i.e., $W_{i}=\{W_{i,j}\ |\ j=1,2,\cdots,F\}$.
\For{$k\in \mathcal{K}$}
\State \begin{eqnarray}
       \label{eq-al-cach}
       \mathcal{Z}_k\leftarrow\{W_{i,j}\ |\ p_{j,k}=*, \forall~i=1,2,\cdots,N\}
       \end{eqnarray}
\EndFor
\EndProcedure
\Procedure{Delivery}{$\mathbf{P}, \mathcal{W},{\bf d}$}
\For{$s=1,2,\cdots,S$}
\State  \begin{eqnarray}\label{eq-coded}
        X_{s}=\bigoplus_{p_{j,k}=s,1\leq j\leq F,1\leq k\leq K}W_{d_{k},j}
        \end{eqnarray}
\State  \begin{eqnarray}\label{eq-coded-divid}
        I_s=\bigcap_{p_{j,k}=s,1\leq j\leq F,1\leq k\leq K}A_k=\{h_{s,1},h_{s,2},\ldots,h_{s,w_s}\},\ \ \ 1\leq w_s<r
        \end{eqnarray}
\State  Divide  $X_{s}$ into $w$ sub-packets, i.e., $X_{s}=\{X_{s,1},\ldots,X_{s,w}\}$
\For{$l=1,2,\cdots,w$}
\State Server sends $X_{s,l}$ to relay $h_l$.
\State Relay $h_l$ sends $X_{s,l}$ to the users which contact with $h_l$
\EndFor
\EndFor
\EndProcedure
\end{algorithmic}
\end{algorithm}
\begin{example}\label{exam-5-10} We can check that the following array is a $(10,5,2,10)$ CPDA.
\begin{small}
\begin{eqnarray}\label{eq-exam-Alg}
\begin{array}{c}
 \ \ \ \ \ \ 123 \ \ \ \ \ 124\ \ \ \ 125\ \ \ \ 134\ \ \ \ 135\ \ \ \ \ 145\ \ \ \ \ 234 \ \ \ \ 235\ \ \ \ \ 245\ \ \ \ \ 345 \\
\mathbf{P}=\left(\begin{array}{cccccccccc}
\ \ 5\ \ &\ \ 6\ \ &\ \ 7\ \ &\ \ 8\ \ &\ \ 9\ \ &\ \ 10\ \ &\ \ *\ \ &\ \ *\ \ &\ \ *\ \ &\ \ \ *\ \ \ \\
2&3&4&*     & *     &*      &8&9&10&*\\
1&*      & *     &3&4&*      &6&7&*      &10\\
 *     &1& *     &2&*      &4&5&*      &7&9\\
 *     & *     &1&*      &2&3&*      &5&6&8
\end{array}\right) \end{array}
\end{eqnarray}
\end{small}
Here each subset is written as a string for short. For instance $\{1,2,3\}$ is written as $123$. Now we will propose a $(5,3,4,10)$ coded caching scheme by $\mathbf{P}$ by Algorithm \ref{alg:CPDA}.
\begin{itemize}
\item Placement phase: From Line 2, each file is split into $F=5$ packets, i.e., $W_n=\{W_{n,j}: j\in [5], n\in [10]\}$. From Line 4, users have the following cache contents respectively:
\begin{eqnarray*}
Z_{123}=\{W_{n,4},W_{n,5}: n\in [N]\}\ \ \ \ \ \ \  Z_{124}=\{W_{n,3},W_{n,5}: n\in [N]\}\\
Z_{125}=\{W_{n,3},W_{n,4}: n\in [N]\}\ \ \ \ \ \ \ Z_{134}=\{W_{n,2},W_{n,5}: n\in [N]\}\\
Z_{135}=\{W_{n,2},W_{n,4}: n\in [N]\}\ \ \ \ \ \ \ Z_{145}=\{W_{n,2},W_{n,3}: n\in [N]\}
\end{eqnarray*}
\item Delivery phase: Assume that users $U_{123}$, $U_{124}$, $U_{125}$, $U_{134}$, $U_{135}$, $U_{145}$, $U_{234}$, $U_{235}$, $U_{245}$, $U_{345}$ request files $W_1$, $W_2$, $W_3$, $W_4$, $W_5$, $W_6$, $W_7$, $W_8$, $W_9$, $W_{10}$ respectively. When $s=1$, by \eqref{eq-coded} the server computes  $$X_1=W_{1,3}\oplus W_{2,4}\oplus W_{3,5}.$$
     By \eqref{eq-coded-divid} we have  $I_1=\{1,2\}$, i.e., the intersection of subsets $\{1,2,3\}$, $\{1,2,4\}$ and $\{1,2,5\}$ which label the first three users. Then the server divides $X_1$ into $|I_1|=2$ sub-packets $X_{1,1}$ and $X_{1,2}$ and sends them to relay $h_1$ and relay $h_2$ respectively. We can see that $U_{123}$ receives subpackets $X_{1,1}$ and $X_{1,2}$ from relay $h_1$ and $h_2$ respectively. It can decode its required packet $W_{1,3}$ from $X_1=(X_{1,1}, X_{1,2})$ since it caches the packets $W_{2,4}$ and $W_{3,5}$. Similarly we can check that user $U_{124}$ can decode required packet $W_{2,4}$ since it can also receive $X_{1,1}$ and $X_{1,2}$ from relays $h_1$ and $h_2$ respectively and caches packets $W_{1,3}$ and $W_{3,5}$, and user $U_{125}$ can decode required packet $W_{3,5}$. We can also check that each user can decode all the requested packets which has not been cached by it. We list all the delivery steps in Table \ref{table-exam-5-10}.
\begin{table}[H]
\centering
\caption{Delivered steps of the $(5,3,4,10)$ coded caching scheme in Example \ref{exam-5-10}.}
\label{table-exam-5-10}
\begin{tabular}{c|c|c|c}
  \hline
  Time slot & Coded Signal & Relays& Transmitted signals \\
  \hline
1 &$X_1=W_{1,3}\oplus W_{2,4}\oplus W_{3,5}$  & $h_1$& $X_{1,1}$\\ \cline{3-4}
 &$=(X_{1,1}\ X_{1,2})$\ \ \ \ \ & $h_2$& $X_{1,2}$\\ \hline
 2 &$X_2=W_{1,2}\oplus W_{4,4}\oplus W_{5,5}$  & $h_1$& $X_{2,1}$\\ \cline{3-4}
 &$=(X_{2,1}\ X_{2,3})$\ \ \ \ \ & $h_3$& $X_{2,3}$\\ \hline
3 &$X_3=W_{2,2}\oplus W_{4,3}\oplus W_{6,5}$  & $h_1$& $X_{3,1}$\\ \cline{3-4}
 &$=(X_{3,1}\ X_{3,4})$\ \ \ \ \ & $h_4$& $X_{3,4}$\\ \hline
4 &$X_4=W_{3,2}\oplus W_{5,3}\oplus W_{6,4}$  & $h_1$& $X_{4,1}$\\ \cline{3-4}
 &$=(X_{4,1}\ X_{4,5})$\ \ \ \ \ & $h_5$& $X_{4,5}$\\ \hline
 5 &$X_5=W_{1,1}\oplus W_{7,4}\oplus W_{8,5}$  & $h_2$& $X_{5,2}$\\ \cline{3-4}
 &$=(X_{5,2}\ X_{5,3})$\ \ \ \ \ & $h_3$& $X_{5,3}$\\ \hline
6 &$X_6=W_{2,1}\oplus W_{7,3}\oplus W_{9,5}$  & $h_2$& $X_{6,2}$\\ \cline{3-4}
 &$=(X_{6,2}\ X_{6,4})$\ \ \ \ \ & $h_4$& $X_{6,4}$\\ \hline
 7 &$X_7=W_{3,1}\oplus W_{8,3}\oplus W_{9,4}$  & $h_2$& $X_{7,2}$\\ \cline{3-4}
 &$=(X_{7,2}\ X_{7,5})$\ \ \ \ \ & $h_5$& $X_{7,5}$\\ \hline
 8 &$X_8=W_{4,1}\oplus W_{7,2}\oplus W_{10,5}$  & $h_3$& $X_{8,3}$\\ \cline{3-4}
 &$=(X_{8,3}\ X_{8,4})$\ \ \ \ \ & $h_4$& $X_{8,4}$\\ \hline
 9 &$X_9=W_{5,1}\oplus W_{8,2}\oplus W_{10,4}$  & $h_3$& $X_{9,3}$\\ \cline{3-4}
 &$=(X_{9,3}\ X_{9,5})$\ \ \ \ \ & $h_5$& $X_{9,5}$\\ \hline
 10 &$X_{10}=W_{6,1}\oplus W_{9,2}\oplus W_{10,3}$  & $h_4$& $X_{10,4}$\\ \cline{3-4}
 &$=(X_{10,4}\ X_{10,5})$\ \ \ \ \ & $h_5$& $X_{10,5}$\\ \hline

\end{tabular}
\end{table}

From Table \ref{table-exam-5-10} each relay transmits $4\times \frac{1}{2}$ packets in all. This implies that the transmission rate
for each relay $h$ is $R_h=\frac{2}{5}$.
\end{itemize}
\end{example}
It is worth noting that the authors in \cite{CYTJ} showed that the array $\mathbf{P}$ in Example \ref{exam-5-10} has the minimum value of $S$ for the fixed $K$, $F$ and $Z$. This implies that for the fixed $K=10$, $F=5$ and $Z=2$, the scheme in Example \ref{exam-5-10} has the minimum transmission rate among the schemes which can be realized based on a CPDA by Algorithm \ref{alg:CPDA}.

\begin{theorem}
\label{th-CPDA-CCS}
For any positive integers $H$, $r$, given a $(K={H\choose r},F,Z,S)$ CPDA, by Algorithm \ref{alg:CPDA}, we can obtain a $(H,r,M,N)$ caching scheme with memory ratio $\frac{M}{N}=\frac{Z}{F}$ such that its packet number is less than or equal to $rF$ and transmission rate for each relay $h$ is
\begin{eqnarray}
\label{eq-rate}
R_{h}=\frac{1}{F}\sum\limits_{s\in[S]}\frac{\psi_h(I_s)}{|I_s|}.
\end{eqnarray}where $I_s$ is defined in \eqref{eq-coded-divid}, and $\psi_h(I_s)=1$ if $h\in I_s$, otherwise $\psi_h(I_s)=0$.
\end{theorem}
\begin{proof}
Given a $({H\choose r},F,Z,S)$ CPDA $\mathbf{P}=(p_{j,k})_{1\leq j\leq F,1\leq k\leq K}$, from Lines 1-6 of Algorithm \ref{alg:CPDA}, we have that each file $W_i\in\mathcal{W}$ is divided into $F$ packets, i.e., $W_{i}=\{W_{i,j}\ |\ j=1,2,\cdots,F\}$, and each user $k$, caches the contents $Z_k=\{W_{i,j}\ |\ p_{j,k}=*, i\in[N]\}$. Since each column has $Z$ stars, user $k$ caches $ZN$ packets. This implies that user $k$ caches $M=\frac{ZN}{F}$ files. So we have $\frac{M}{N}=\frac{Z}{F}$.

For any request vector ${\bf d}=(d_1,d_2,\ldots,d_K)$, let us consider delivery phase in Lines 7-17 of Algorithm \ref{alg:CPDA}. For any packet, say $W_{d_k,j}$, which is requested by user $k$ but has not been cached, by \eqref{eq-al-cach} the entry $p_{j,k}$ must be an integer, denoted by $s$. Assume that there are $g\geq 1$ entries such that $p_{j_1,k_1}=p_{j_2,k_2}=\ldots=p_{j_g,k_g}=p_{j,k}=s$ where $j_l\in [F]$, $k_l\in \mathcal{K}$, and $j_l\neq j$, $k_l\neq k$, $l\in[g]$. By \eqref{eq-coded}, the server computes the coded signal $$X_{s}=\left(\bigoplus_{l=1}^{g}W_{d_{k_l},j_l}\right)\bigoplus W_{d_k,j}.$$
From the property C2-b) of a PDA, we have that $p_{k,j_l}=*$. This implies that user $k$ has cached the packets $W_{d_{k_l},j_l}$ for all $l\in[g]$. So it can decode its requested packet $W_{d_k,j}$ directly if it can receive the coded signal $X_{s}$.

Assume that these column numbers $k_l$ and $k$ are represented by $r$-subsets $A_l$, $l\in [g]$, and $A$. By \eqref{eq-coded-divid}, we have $$I_s=\left(\bigcap_{l=1}^{g}A_l\right)\bigcap A=\{h_{s,1},h_{s,2},\ldots,h_{s,w_s}\}$$ for some positive integer $w_s$. Then the server divides $X_{s}$ into $w_s$ sub-packets, i.e., $X_{s}=\{X_{s,1},\ldots,X_{s,w_s}\}$, and sends $X_{s,i}$ to relay $h_{s,i}$ respectively, $i\in[w_s]$. Since $|I_s|\leq r$ always holds, the packet number is at most $rF$. Finally each relay $h_{s,i}$ transmits $X_{s,i}$ to user $k$ since $k\in \mathcal{H}_{h_{s,i}}$. Clearly user $k$ can receive all the sub-coded signals since $I_s\subseteq A$. Then it can obtain the coded signal $X_s$.

From the above discussion, we have proved that by Algorithm \ref{alg:CPDA}, each user can decode its requesting files. For each relay $h$ and for any integer $s$ occurring in $\mathbf{P}$, if $h\in I_s$ the relay $h$ transmits the sub-coded signal with size $\frac{1}{|I_s|}$ packets by the definition of function $\psi_h(I_s)$. So the total amount of data transmitted by relay $h$ is $\sum\limits_{s\in[S]}\frac{\psi_h(I_s)}{|I_s|}$ packets. Then the transmission rate for relay $h$ in \eqref{eq-rate} can be directly obtained.
\end{proof}

\section{Placement delivery array for combination networks}
\label{sec-new-scheme}
\subsection{Some known combinatorial PDAs}
\label{subsec-strong}
In this subsection, we will show that the PDAs in Lemma \ref{le:subset} are also CPDAs. For the readers' convenience, here we include a short introduction of construction of PDAs in \cite{Jen1997subset,YTCC}. We should point out that we modify the representation of the alphabet of the PDA in \cite{YTCC} for our following proofs.
\begin{construction}(\cite{Jen1997subset,YTCC})
\label{constrct1}For any positive integers $H$, $b$, $r$, $\lambda$ satisfying $0<r,b<H$, $\lambda<\min\{r,b\}$, define
\begin{eqnarray*}
\mathcal{F}={[H]\choose b},\ \  \mathcal{K}={[H]\choose r}\ \ \mathcal{I}={[H]\choose \lambda}\ \ \hbox{and}\ \ \mathcal{I}'={[H]\choose r-\lambda}.
\end{eqnarray*} Then we can obtain
\begin{itemize}
\item a $({H\choose r},{H\choose b},{H\choose b}-{r\choose \lambda}{H-r\choose b-\lambda},{H\choose r+b-2\lambda}{H-r-b+2\lambda\choose \lambda})$ PDA $\mathbf{P}=(p_{B,A})_{B\in \mathcal{F}, A\in \mathcal{K}}$ and
\item a $({H\choose r},{H\choose b},{H\choose b}-{r\choose \lambda}{H-r\choose b-\lambda},{H\choose r+b-2\lambda}{r+b-2\lambda\choose r-\lambda})$ PDA $\mathbf{P}'=(p'_{B,A})_{B\in \mathcal{F}, A\in \mathcal{K}}$
\end{itemize} where
\begin{eqnarray}
\label{eq-rule}
\begin{split}
p_{B,A}=\left\{
\begin{array}{cc}
((A\cup B)-I,I)& \hbox{if} \ A\cap B=I\in \mathcal{I}\\
*&\hbox{Otherwise}
\end{array}
\right.\ \ \ \ \ \ \ \ \ \ \ \ \ \ \hbox{and}\\[0.3cm]
p'_{B,A}=\left\{
\begin{array}{cc}
((A\cup B)-I,A-B)& \hbox{if} \ A\cap B=I\in \mathcal{I}\\
*&\hbox{Otherwise}
\end{array}
\right.\ \ \ \ \ \ \ \ \ \ \ \ \ \
\end{split}
\end{eqnarray}
\end{construction}

\begin{example} When $H=5$, $r=3$, $b=1$, $\lambda=1$, from Construction \ref{constrct1}, we have the following two PDAs
\begin{small}
\begin{eqnarray*}
\label{eq-exam-1}
\begin{array}{c}
\ \ \ 123 \ \ \ \  \ \ \ 124\ \ \ \  \ \ \  125\ \ \ \  \ \ \ 134\ \ \ \  \ \ \ 135\ \ \ \  \ \  145\ \ \ \  \ \ \ 234\ \ \ \  \ \ \ 235\ \ \ \  \ \ \ 245\ \ \ \  \ \ \ 345 \\
\mathbf{P}=\left(\begin{array}{cccccccccc}
(23,1)&(24,1)&(25,1)&(34,1)&(35,1)&(45,1)&*     &*     &*     &*\\
(13,2)&(14,2)&(15,2)&*     & *    &*     &(34,2)&(35,2)&(45,2)&*\\
(12,3)&*     & *    &(14,3)&(15,3)&*     &(24,3)&(25,3)&*     &(45,3)\\
 *    &(12,4)& *    &(13,4)&*     &(15,4)&(23,4)&*     &(25,4)&(35,4)\\
 *    & *    &(12,5)&*     &(13,5)&(14,5)&*     &(23,5)&(45,2)&(34,5)
\end{array}\right) \begin{array}{c}
1\\
2\\
3\\
4\\
5
\end{array}\end{array}
\end{eqnarray*}\end{small}
\begin{small}
\begin{eqnarray}\label{eq-exam-2}
\begin{array}{c}
\ 123 \ \ \ \ \ \ \ \ 124\ \ \ \ \ \ \ \ 125\ \ \ \ \ \ \ \ \ 134\ \ \ \ \ \ \ \ 135\ \ \ \ \ \ \ \  145\ \ \ \ \ \ \ \ 234\ \ \ \ \ \ \ \ 235\ \ \ \ \ \ \ \ 245\ \ \ \ \ \ \ \ 345 \\
\mathbf{P}'=\left(\begin{array}{cccccccccc}
(23,23)&(24,24)&(25,25)&(34,34)&(35,35)&(45,45)&*      &*      &*      &*\\
(13,13)&(14,14)&(15,15)&*      & *     &*      &(34,34)&(35,35)&(45,45)&*\\
(12,12)&*      & *     &(14,14)&(15,15)&*      &(24,24)&(25,25)&*      &(45,45)\\
 *     &(12,12)& *     &(13,13)&*      &(15,15)&(23,23)&*      &(25,25)&(35,35)\\
 *     & *     &(12,12)&*      &(13,13)&(14,14)&*      &(23,23)&(24,24)&(34,34)
\end{array}\right) \begin{array}{c}
1\\
2\\
3\\
4\\
5
\end{array}\end{array}
\end{eqnarray}
\end{small}
Here each subset is written as a string for short. For instance $\{1,2,3\}$ is written as $123$. Replacing $(12,12)$, $(13,13)$, $(14,14)$, $(15,15)$, $(23,23)$, $(24,24)$, $(25,25)$, $(34,34)$, $(35,35)$ and $(45,45)$ by $1$, $2$, $\ldots$, $10$ respectively, the arrays in \eqref{eq-exam-2} and \eqref{eq-exam-Alg} are the same.
\end{example}
\begin{theorem}
\label{th-main-1}
For any positive integers $H$, $r$, $b$, $\lambda$ satisfying $0<r,b<H$, $\lambda\leq\min\{r,b\}$ and $r+b-2\lambda<H$, the PDAs generated by Construction \ref{constrct1} is a CPDA which gives a $(H,r,M,N)$ coded caching scheme  with memory ratio $\frac{M}{N}=1-\frac{{r\choose \lambda}{H-r\choose b-\lambda}}{{H\choose b}}$, transmission rate $R=\frac{{H\choose r+b-2\lambda}}{H{H\choose b}}\cdot\min\left\{{H-(r+b-2\lambda)\choose \lambda},{r+b-2\lambda\choose r-\lambda }\right\}$ and packet number $F<r{H\choose b}$.
\end{theorem}
\begin{proof}
Let us consider $\mathbf{P}$ in Construction \ref{constrct1} first. Since $\mathbf{P}$ is already a PDA, in order to show that $\mathbf{P}$ is CPDA, we only need to consider property C3. For any set system $(C,I)$ in $\mathbf{P}$, by \eqref{eq-rule} if column $A$ contains $(C,I)$, then $I\subseteq A$ always holds since $I=A\cap B$ for some $b$-subset $B$. So $\mathbf{P}$ satisfies C3. Similarly we can also show $\mathbf{P}'$ satisfies C3 directly.

From Theorem \ref{th-CPDA-CCS}, the scheme realized by $\mathbf{P}$ has memory ratio $\frac{M}{N}=\frac{Z}{F}=1-\frac{{r\choose \lambda}{H-r\choose b-\lambda}}{{H\choose b}}$. Now let us consider its  packet number and transmission rate. From the first formula in \eqref{eq-rule}, we know that for any set system $(C,I)$ occurring in $\mathbf{P}$, all the $r$-subsets, each of which is used to represent a column containing $(C,I)$, are $\mathcal{K}_{C,I}=\{A'\bigcup I| A'\in{C\choose r-\lambda}\}$. This implies $\bigcap_{A\in \mathcal{K}_{C,I}}A=I$ since $\bigcap_{A'\in{C\choose r-\lambda}}A'=\emptyset$. So from Lines 10-11 in Algorithm \ref{alg:CPDA}, the packet number is $F=\lambda{H\choose b}$. Furthermore, for any relay $h$, there are exactly ${H-1\choose r+b-2\lambda}{H-(r+b-2\lambda)-1\choose \lambda -1}$ set system $(C,I)$ such that $h\in I$. So from Theorem \ref{th-CPDA-CCS}, we have the transmission rate for relay $h$ is
\begin{eqnarray*}
R_{h}&=&\frac{1}{F}\sum\limits_{s\in[S]}\frac{\psi_h(I_s)}{|I_s|}=\frac{1}{F}{H-1\choose r+b-2\lambda}{H-(r+b-2\lambda)-1\choose \lambda -1}\\
&=& \frac{1}{F}\frac{\lambda}{H}{H\choose r+b-2\lambda}{H-(r+b-2\lambda)\choose \lambda}\\
&=& \frac{{H\choose r+b-2\lambda}{H-(r+b-2\lambda)\choose \lambda}}{H{H\choose b}}.
\end{eqnarray*}

Similarly we can also show that the scheme realized by $\mathbf{P}'$ has memory ratio $\frac{M}{N}=\frac{Z}{F}=1-\frac{{r\choose \lambda}{H-r\choose b-\lambda}}{{H\choose b}}$, packet number $F=(r-\lambda){H\choose b}$ and transmission rate for each relay $R_h=\frac{{H\choose r+b-2\lambda}{r+b-2\lambda\choose r-\lambda}}{H{H\choose b}}$. Since $1\leq \lambda, r-\lambda<r$, two schemes realized respectively by $\mathbf{P}$ and  $\mathbf{P}'$ have packet number $F<r{H\choose b}$. We can get the minimum transmission rate between these two schemes $\frac{{H\choose r+b-2\lambda}{H-(r+b-2\lambda)\choose \lambda}}{H{H\choose b}}$ and $\frac{{H\choose r+b-2\lambda}{r+b-2\lambda\choose r-\lambda}}{H{H\choose b}}$.

Then the proof is completed.
\end{proof}
\begin{remark}
\label{remark-strong}
For any positive integer $H$ and $1\leq r<H$, let $b=\lambda=r-1$ in Lemma \ref{le:subset}. Then a $({H\choose r}, {H\choose r-1}, {H\choose r-1}-r,H)$ PDA can be obtained. That is exactly the PDA in the third PDA in Lemma \ref{le-MN-C}.  From Remark \ref{remark-1}, we know such a PDA has the minimum integer $S$ for the fixed values of $K$, $F$ and $Z$. From Theorem \ref{th-main-1}, such a PDA is also a CPDA which is the result in Lemma \ref{le-yan-2}.
\end{remark}
\subsection{Generalized constructions}
\label{sec-genneralzied}
In this subsection, we will propose a new construction of CPDAs which can be regarded as the extension of Construction \ref{constrct1}. The new CPDAs could generate the schemes with flexible memory ratios.
\begin{construction}
\label{constrct2}For any positive integers $H$, $b$, $r$ and $\lambda$ with $b<r+\lambda<H$ and $\lambda<b$, let
\begin{eqnarray*}
\mathcal{F}=\left\{(B, \Gamma)\ |\  B\in{[H]\choose b}, \Gamma\subseteq B, \Gamma\in{[H]\choose \lambda}\right\},\ \ \ \mathcal{K}={[H]\choose r}.
\end{eqnarray*}
Clearly $|\mathcal{F}|={H\choose b}{b\choose \lambda}$ since $|{[H]\choose b}|={H\choose b}$ and  $|{B\choose \lambda}|={b\choose \lambda}$. Define a ${H\choose b}{b\choose \lambda}\times{H\choose r}$ array $\mathbf{P}=(p_{(B,\Gamma),A})_{(B,\Gamma)\in \mathcal{F},A\in \mathcal{K}}$ in the following rule.
\begin{eqnarray}
\label{eq-rule-2}
p_{(B,\Gamma),A}=\left\{
\begin{array}{cc}
((A\cup \Gamma),A-B)& \hbox{if} \ A\cap \Gamma=\emptyset,\  B\subseteq A\cup \Gamma\\
*&\hbox{Otherwise}
\end{array}
\right.
\end{eqnarray}
\end{construction}
\begin{theorem}
\label{th-main-2}For any positive integers $H$, $b$, $r$ and $\lambda$ with $b<r+\lambda<H$ and $\lambda<b$, there exists a $({H\choose r}, {H\choose b}{b\choose \lambda},{H\choose b}{b\choose \lambda}- {H-r\choose \lambda}{r\choose b-\lambda}, {H\choose\lambda+r}{r+\lambda\choose r+\lambda-b})$ CPDA which gives a $(H,r,M,N)$ coded caching scheme with memory ratio $\frac{M}{N}=1-\frac{{H-r\choose \lambda}{r\choose b-\lambda}}{{H\choose b}{b\choose \lambda}}$, transmission rate for each relay $R_h=\frac{{H\choose\lambda+r}{r+\lambda\choose r+\lambda-b}}{H{H\choose b}{b\choose \lambda}}$ and  packet number $F=(r+\lambda-b){H\choose b}{b\choose \lambda}$.
\end{theorem}
\begin{proof}
Now let us consider the conditions of a CPDA respectively. By \eqref{eq-rule-2} we have that each column has ${H-r\choose \lambda}{r\choose b-\lambda}$ nonstar entries, i.e., $Z={H\choose b}{b\choose \lambda}- {H-r\choose \lambda}{r\choose b-\lambda}$. Furthermore, we claim that C2 holds. For any set system $(B,\Gamma)\in \mathcal{F}$ and $A\in \mathcal{K}$, if $A\cap \Gamma=\emptyset$ then $A-B=(A\cup \Gamma)-B$. Assume that there are two different entries say $p_{(B_1,\Gamma_1),A_1}=p_{(B_2,\Gamma_2),A_2}=(C,I)\in {[H]\choose r+\lambda}\times {[H]\choose r+\lambda-b}$. Then we have
\begin{eqnarray}
\label{eq-constr-2}
C=A_1\cup \Gamma_1=A_2\cup \Gamma_2,\ \ \ \ \ \ \ I=A_1-B_1=A_2-B_1
\end{eqnarray}
If $A_1=A_2$, we have $\Gamma_1=\Gamma_2$ by the first item of \eqref{eq-constr-2} and $B_1=B_1$ by the second item of \eqref{eq-constr-2}. This contradicts our hypothesis. Similarly if $(B_1,\Gamma_1)=(B_2,\Gamma_2)$, we have $A_1=A_2$ by the first item of \eqref{eq-constr-2}. This is impossible by our hypothesis. If $A_1\neq A_2$ and $(B_1,\Gamma_1)\neq(B_2,\Gamma_2)$, we have $A_1\cap \Gamma_2\neq \emptyset$ and $A_2\cap \Gamma_1\neq\emptyset$ by the first item of \eqref{eq-constr-2}. Then $p_{(B_1,\Gamma_1),A_2}=p_{(B_2,\Gamma_2),A_1}=*$. We can easily check that C3 holds since any column (which is indexed by $A\in \mathcal{K}$) containing set system $(C,I)$ if and only if $I\subseteq A$ by \eqref{eq-rule-2}.
Furthermore there are $S={H\choose r+\lambda}{r+\lambda\choose r+\lambda-b}$ different set systems in $\mathbf{P}$. Then we have a $({H\choose r}, {H\choose b}{b\choose \lambda},{H\choose b}{b\choose \lambda}- {H-r\choose \lambda}{r\choose b-\lambda}, {H\choose\lambda+r}{r+\lambda\choose r+\lambda-b})$ CPDA.

Given a set system $(C,I)\in {[H]\choose r+\lambda}\times {[H]\choose r+\lambda-b}$ occurring in $\mathbf{P}$, each $r$-subset $A'\in {C-I\choose b-\lambda}$ such that $A=A'\cup I$ representing a column contains $(C,I)$. Then the intersection of such $A$s is $I$ since the intersection of all the subsets in ${C-I\choose b-\lambda}$ is empty set. In addition, for each relay $h$, there are exactly ${H-1\choose r+\lambda-1}{r+\lambda-1\choose r+\lambda-b-1}$ set systems occurring in $\mathbf{P}$ contains relay $h$. So by \eqref{eq-rate}
\begin{eqnarray*}
R_{h}=\frac{1}{(r+\lambda -b){H\choose b}{b\choose \lambda}}{H-1\choose r+\lambda-1}{r+\lambda-1\choose r+\lambda-b-1}=\frac{{H\choose r+\lambda}{r+\lambda\choose r+\lambda-b}}{H{H\choose b}{b\choose \lambda}}
\end{eqnarray*}
\end{proof}
\section{Performance analyses}
\label{sec-peformance}
From Remark \ref{remark-strong}, we know that the result in Lemma \ref{le-yan-2} is a special case of the schemes in Theorem \ref{th-main-1}. From Lemma \ref{le-yan-recursive}, we know that the schemes in \cite{Yan} have a strong limitation, i.e., $r|H$. However our schemes in  Theorems \ref{th-main-1} and \ref{th-main-2} have no such a limitation. Furthermore, our schemes have significant advantages on the packet number than the schemes in Corollary \ref{co-yan-1} by sacrificing a little transmission rate, and advantages on the packet number and the transmission rate than that of the schemes generated by Lemmas \ref{le:subset} and \ref{le-yan-recursive}, i.e., the scheme obtained by grouping method. Since the parameters in Theorem  \ref{th-main-1} and \ref{th-main-2} are too complex, here we just list two examples to show that our schemes have significant advantages on the transmission rates or packet numbers.

When $H=20$, $r=4$, we have $K={20\choose 4}=4845$. From Theorems \ref{th-main-1} and \ref{th-main-2} we have two schemes respectively. Choosing the scheme with smaller rate of these two schemes, a new scheme, say Scheme1, can be obtained. We list its transmission rate and the packet number in blue lines of Figures \ref{fig.ny-R} and \ref{fig.ny-F} respectively.
\begin{figure}[h]
  \centering
  \includegraphics[width=10cm,height=6cm]{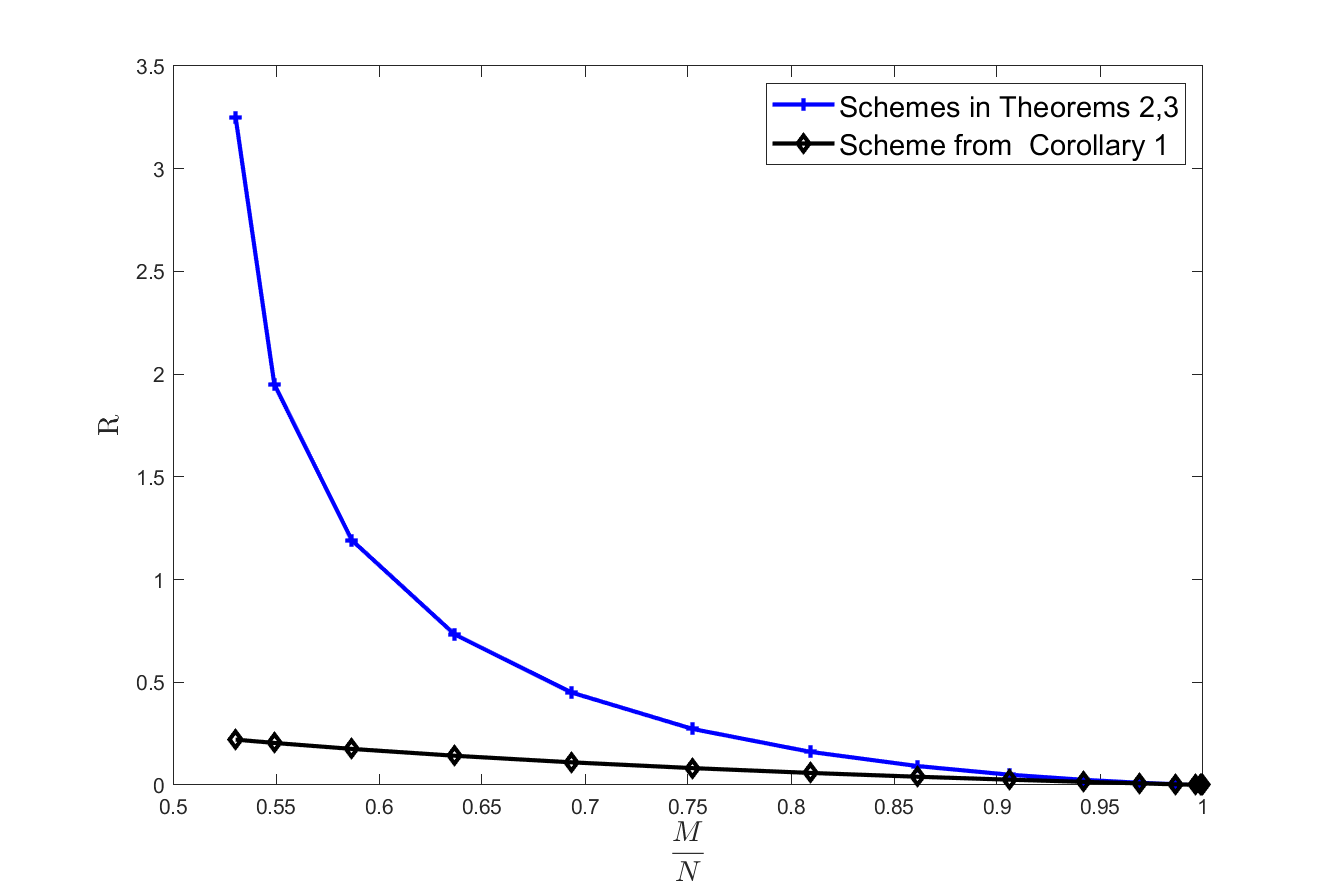}
  \caption{The transmissions for each relay of our new schemes and the scheme from Corollary \ref{co-yan-1}}\label{fig.ny-R}
\end{figure}
\begin{figure}[h]
  \centering
  \includegraphics[width=10cm,height=6cm]{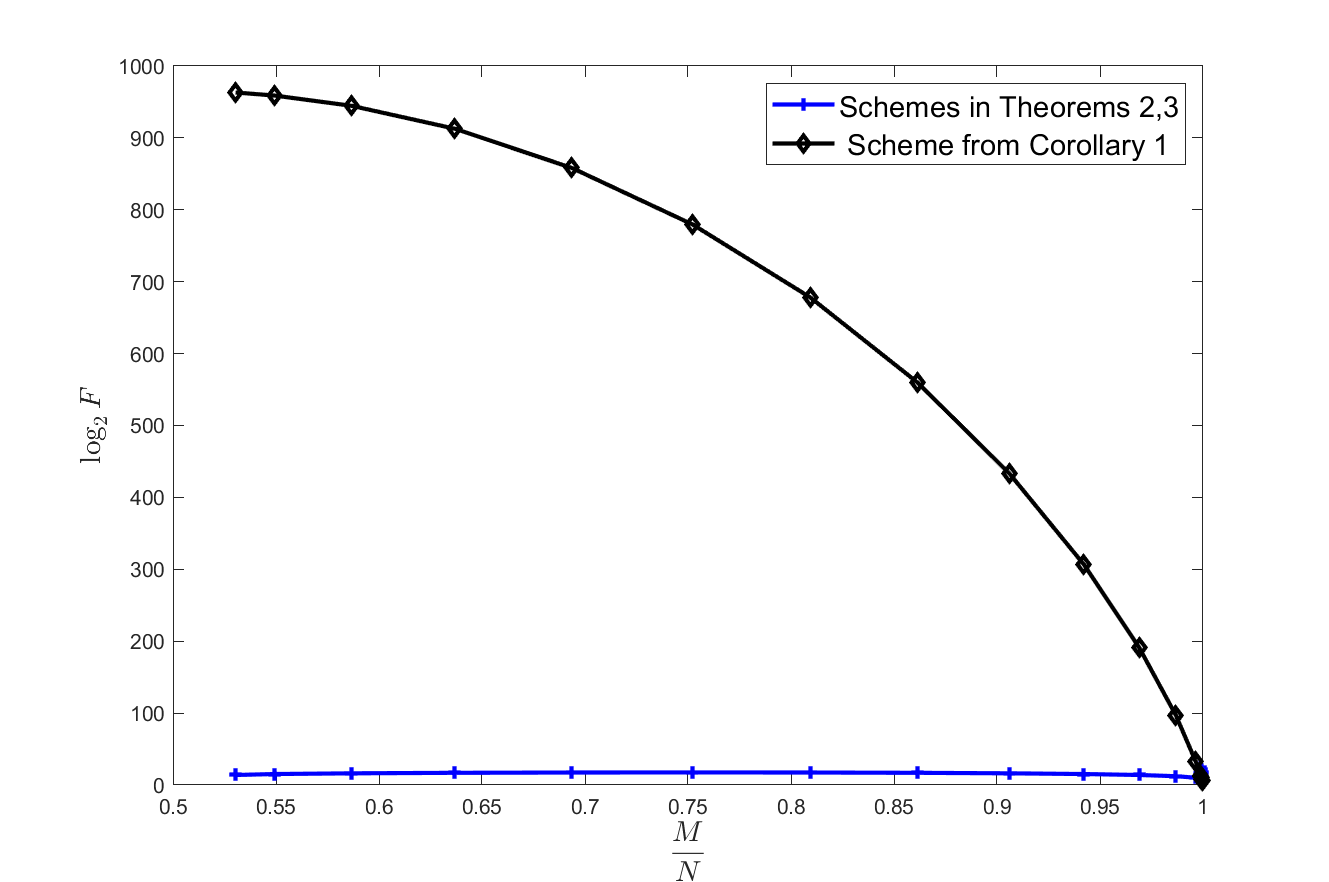}
  \caption{The packet numbers of our new schemes and the scheme from Corollary \ref{co-yan-1}}\label{fig.ny-F}
\end{figure}
From Corollary \ref{co-yan-1}, a scheme, say Scheme2, can be obtained. We list its transmission rate and the packet number in black lines of Figures \ref{fig.ny-R} and \ref{fig.ny-F} respectively. From Figure \ref{fig.ny-R} and \ref{fig.ny-R}, the transmission rate of Scheme1 has just several times larger than that of Scheme2 but the packet number of Scheme1 is far less than that of Scheme2.

Given the $({H-1\choose r-1},{H-1\choose b}, {H-1\choose b}-{r-1\choose \lambda}{H-r\choose b-\lambda},S)$ PDA with $S={H-1\choose r-1+b-2\lambda}\cdot\min\left\{{H-(r+b-2\lambda)\choose \lambda},{r-1+b-2\lambda\choose r-1-\lambda }\right\}$ in Lemma \ref{le:subset}, using Lemma \ref{le-yan-recursive} we have a $({H\choose r},r{H-1\choose b}, r{H-1\choose b}-r{r-1\choose \lambda}{H-r\choose b-\lambda},HS)$ CPDA. From Lemma \ref{le-CPDA-CC}, we have a coded caching scheme, Scheme3, where the transmission rate and packet number are listed in red lines in Figures \ref{fig.nys-R} and \ref{fig.nys-F}. From these two figures, we can see that our scheme has smaller transmission rates for each relay (the blue line in Figure \ref{fig.nys-R}) and packet numbers (the blue line in Figure \ref{fig.nys-F}) than that of Scheme3. So our schemes have better performance than Scheme3.
\begin{figure}[h]
  \centering
  \includegraphics[width=10cm,height=6cm]{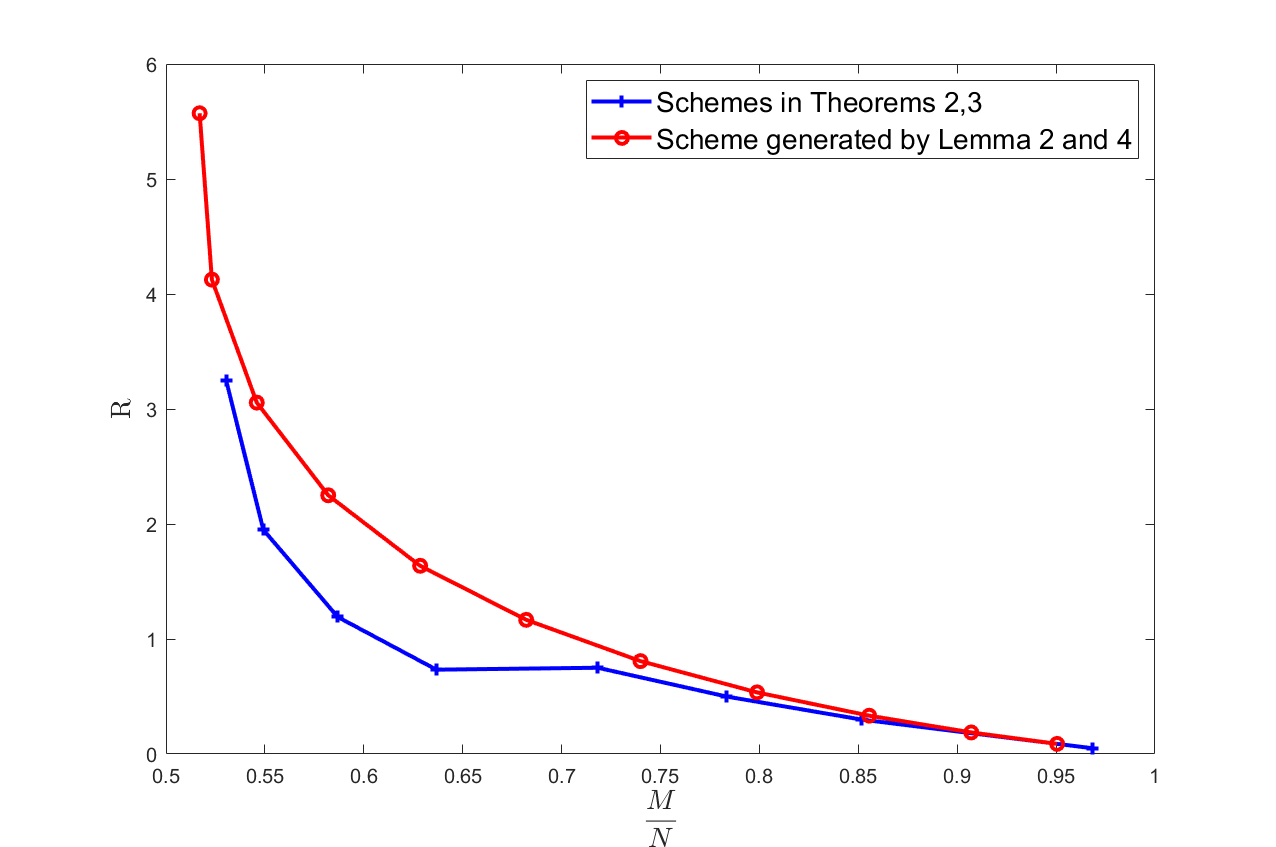}
  \caption{The transmissions for each relay of our new schemes and the scheme generated by Lemma \ref{le:subset} and \ref{le-yan-recursive}}\label{fig.nys-R}
\end{figure}
\begin{figure}[h]
  \centering
  \includegraphics[width=10cm,height=6cm]{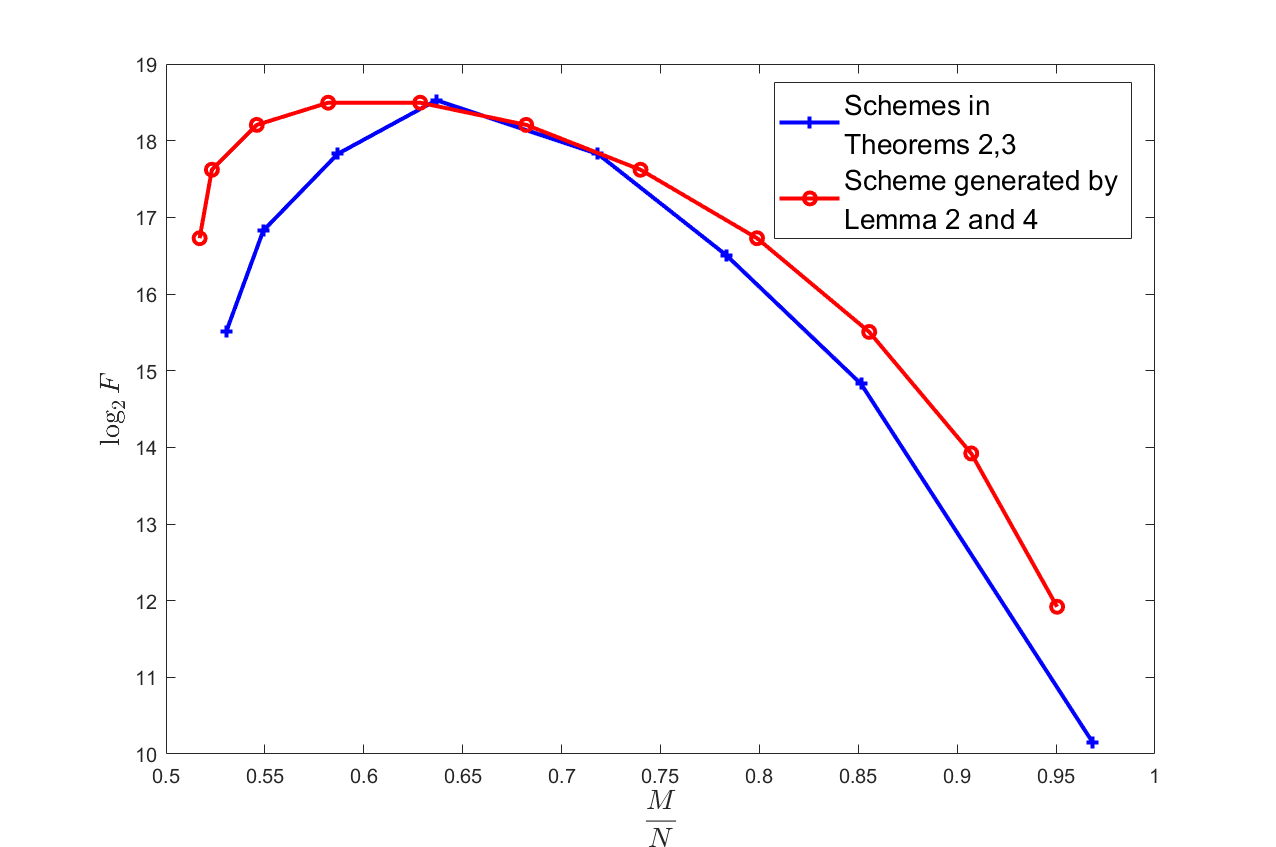}
  \caption{The packet numbers of our new schemes and the scheme generated by Lemma \ref{le:subset} and \ref{le-yan-recursive}}\label{fig.nys-F}
\end{figure}
\section{Conclusion}
\label{conclusion}
In this paper, a new algorithm, which can be used to realize a coded caching scheme for combination network based on a CPDA, was proposed. Consequently the schemes obtained by our algorithm have smaller packet number and are implemented more flexible than the previously known results. Then we focused on directly constructing CPDAs. First we showed that the known PDAs in \cite{YTCC} are exactly CPDAs, and proposed a new CPDAs which can be regarded as an extension of the PDAs in \cite{YTCC}. It is worth noting that our CPDAs hold for any positive integers $H$ and $r$ with $r<H$. Finally we showed that our schemes have better performance than that of the schemes in \cite{Yan}. So considering the direct construction of CPDAs is
interesting and useful.

\end{document}